\documentclass[conference]{IEEEtran}
\IEEEoverridecommandlockouts
\usepackage[utf8]{inputenc}
\usepackage{cite}
\usepackage{overpic}
\usepackage{amsmath,amssymb,amsfonts}
\usepackage{bm}
\usepackage{graphicx}
\usepackage{textcomp}
\usepackage{xcolor}
\usepackage{float}
\usepackage{amsthm}
\usepackage{graphicx}
\usepackage{epstopdf}
\usepackage{amsmath,bm,bbm}
\usepackage{amsfonts}
\usepackage{amssymb}
\usepackage{color}
\usepackage{multirow}
\usepackage{multicol}
\usepackage{soul,xcolor}
\usepackage{setspace}
\newtheorem{lemma}{Lemma}

\usepackage{longtable}
\usepackage{mathtools}
\usepackage{subfig}
\usepackage{tabulary}
\usepackage{epsfig}
\usepackage{caption}
\usepackage{booktabs}
\usepackage{blindtext}
\usepackage{adjustbox}
\usepackage{hyperref}
\usepackage{dirtytalk}
\usepackage{comment} 
\usepackage{tikz}
\usetikzlibrary{positioning} 
\usepackage{algorithm}

\usepackage{algpseudocode}

\definecolor{mypurple}{HTML}{9933FF}
\definecolor{mygreen}{HTML}{009900}

\DeclareMathOperator{\tr}{tr}

\theoremstyle{plain}

\def\diag{\mathrm{diag}}

\def\tr{\mathrm{tr}}

\def\Htran{\mbox{\tiny $\mathrm{H}$}}
\def\Ttran{\mbox{\tiny $\mathrm{T}$}}
 
\def\imagunit{\mathsf{j}} 
\begin{document}
\bstctlcite{IEEEexample:BSTcontrol}
\makeatletter
\newcommand*{\rom}[1]{\expandafter\@slowromancap\romannumeral #1@}
\makeatother

\title{ Dual-polarized RIS versus dual-polarized network-controlled repeater  \vspace{-0.4cm}  }

\author{\IEEEauthorblockN{\"Ozlem Tu\u{g}fe Demir}
\IEEEauthorblockA{\textit{Department of Electrical and Electronics Engineering} \\
\textit{Bilkent University}\\
Ankara, Turkiye \\
E-mail: ozlemtugfedemir@bilkent.edu.tr} 
 \vspace{-8mm}
}
\maketitle
\begin{abstract}
We study the achievable rate performance of dual-polarized passive reconfigurable intelligent surfaces (RISs) and dual-polarized network-controlled repeaters (NCRs) in a point-to-point wireless link. While RISs rely on passive beamforming with large array gains, NCRs exploit active amplification across polarization branches. We develop a unified analytical framework that captures the impact of amplification, noise propagation, and deployment geometry under different power-budget models. Closed-form rate expressions are derived, and the optimal one-stream and two-stream transmission strategies for NCRs are characterized. The results reveal that RISs are competitive in short-range scenarios and benefit from large array sizes, whereas NCRs become preferable in long-range deployments due to their ability to compensate for pathloss through amplification.
\end{abstract}

\vspace{-2mm}

\section{Introduction}

Reconfigurable intelligent surfaces (RISs) have emerged as a promising technology for shaping the wireless propagation environment in a controllable manner~\cite{liu2021reconfigurable}. By applying programmable phase shifts, RISs can enhance signal propagation, particularly in scenarios where the direct link is blocked or severely attenuated. Due to their nearly passive nature, RISs offer a cost- and energy-efficient solution for coverage enhancement, albeit at the expense of requiring large surface areas to compensate for the inherent double-cascaded pathloss~\cite{bjornson2021reconfigurable}. 

Recent works have considered dual-polarized RIS architectures in different contexts \cite{ramezani2023dual,han2021dual}. By exploiting polarization diversity, one can improve spectral efficiency. In particular, leveraging orthogonal polarization branches provides additional degrees of freedom compared to conventional single-polarized systems, which is especially beneficial in line-of-sight (LOS) dominated environments.

In parallel, network-controlled repeaters (NCRs) have gained renewed attention as a practical and standardized solution for coverage extension in beyond-5G systems~\cite{wen2024shaping,willhammar2025achieving,demir2026ncr,topal2025fair}. NCRs operate by amplifying and forwarding the received signal, including noise, with minimal processing delay, thereby acting as active relays. 

Despite the growing interest in both RIS- and NCR-assisted communications, existing studies have largely focused on single-polarized models or purely numerical comparisons. In particular, while dual-polarized RIS systems have been investigated, a corresponding dual-polarized analysis for NCR-assisted links is still missing. To the best of our knowledge, this paper provides the first analytical study of dual-polarized NCR-assisted communications and establishes a unified framework for comparing it with dual-polarized RIS systems.

The main contributions of this paper are summarized as follows:
\begin{itemize}
    \item We propose a unified dual-polarized system model for RIS- and NCR-assisted communications and derive closed-form achievable rate expressions for both technologies.
    
    \item We show that, depending on the power constraints and propagation conditions, the optimal NCR strategy can switch between one-stream and two-stream transmission across polarization branches.
    
    \item We compare joint and separate power-budget formulations for NCRs and demonstrate how the additional flexibility in power allocation affects performance.
    
    \item We identify geometry-dependent operating regimes where passive RIS or NCR is preferable, thereby providing design insights for practical deployments.
\end{itemize}
\vspace{-2mm}
\section{System Model and Achievable Rates}\label{sec2}
\vspace{-2mm}

We consider an uplink system where the UE is equipped with a dual-polarized antenna, while the BS employs $M$ dual-polarized antennas. The two branches, indexed by $H$ and $V$, correspond to two orthogonal polarizations. Communication is assisted either by a passive RIS or by an NCR. We assume that the direct UE--BS link is blocked. The case where a direct link exists will be considered in the extended version of this paper.

We let the transmitted symbol vector for the two polarizations be $\mathbf{s}=
    \begin{bmatrix}
        s_H& s_V
    \end{bmatrix}^{\Ttran}\in\mathbb{C}^{2},
    \qquad
    \mathbb{E}\{\mathbf{s}\mathbf{s}^{\Htran}\}=\mathbf{Q}$,
where $\mathbf{Q}\succeq \mathbf{0}$ and $\tr(\mathbf{Q})\le P$.

For the RIS, we adopt a dual-polarized uniform linear array (ULA) model with $N$ elements per polarization \cite{ramezani2023dual}. The corresponding array response vectors are
\begin{align}
   & \mathbf{a}_{{\rm RIS},H}(\phi)\nonumber\\
    &=
    \begin{bmatrix}
        1 &
        e^{-\imagunit\frac{2\pi}{\lambda}(2\Delta_d)\sin(\phi)} &
        \cdots &
        e^{-\imagunit(N-1)\frac{2\pi}{\lambda}(2\Delta_d)\sin(\phi)}
    \end{bmatrix}^{\Ttran},
    \\
    &\mathbf{a}_{{\rm RIS},V}(\phi)
    =
    e^{-\imagunit\frac{2\pi}{\lambda}\Delta_d\sin(\phi)}
    \mathbf{a}_{{\rm RIS},H}(\phi),
\end{align}
where $\Delta_d$ is the inter-element distance and the $V$ and $H$ polarizations are stacked side by side (i.e., concatenated), as shown in \cite{ramezani2023dual}. The BS-side array responses $\mathbf{a}_{{\rm BS},H}(\phi)$ and $\mathbf{a}_{{\rm BS},V}(\phi)$ are defined similarly. 

\vspace{-2mm}
\subsection{Passive RIS-Assisted Communications}
\vspace{-1mm}

The UE--RIS channel is modeled as
\vspace{-2mm}
\begin{align}
    \mathbf{H}_1
    =
    \sqrt{\beta_1}e^{-\imagunit\theta_1}
    \begin{bmatrix}
        \mathbf{a}_{{\rm RIS},H}(\tilde\phi) & \mathbf{0}\\
        \mathbf{0} & \mathbf{a}_{{\rm RIS},V}(\tilde\phi)
    \end{bmatrix},
\end{align}
where $\beta_1>0$ is the large-scale channel gain, and $\theta_1$ represents the phase shift associated with the first RIS element. Moreover, $\tilde{\phi}$ denotes the angle-of-arrival (AoA) of the UE. The RIS--BS channel is
\vspace{-2mm}
\begin{align}
    \mathbf{H}_2
    =
    \sqrt{\beta_2}e^{-\imagunit\theta_2}
    \begin{bmatrix}
        \mathbf{a}_{{\rm BS},H}(\phi)\mathbf{a}_{{\rm RIS},H}^{\Ttran}(\phi) & \mathbf{0}\\
        \mathbf{0} & \mathbf{a}_{{\rm BS},V}(\phi)\mathbf{a}_{{\rm RIS},V}^{\Ttran}(\phi)
    \end{bmatrix},
\end{align}
where $\beta_2 > 0$ denotes the large-scale channel gain of the RIS--BS link, and $\theta_2$ represents the phase shift associated with the first RIS element and first BS antenna along this link. Moreover, $\phi$ denotes the angle-of-departure (AoD) from the RIS and, equivalently, the AoA at the BS. 
The RIS phase-shift matrix is
\vspace{-2mm}
\begin{align}
    \mathbf{\Psi}_{\rm P}
    =
    \diag(\boldsymbol{\phi}_H,\boldsymbol{\phi}_V),
\end{align}
where $\boldsymbol{\phi}_H,\boldsymbol{\phi}_V\in\mathbb{C}^N$ have unit-modulus entries.

The received signal at the BS is
\vspace{-2mm}
\begin{align}
    \mathbf{y}
    =
    \mathbf{H}_2\mathbf{\Psi}_{\rm P}\mathbf{H}_1\mathbf{s}
    +\mathbf{n},
    \qquad
    \mathbf{n}\sim\mathcal{CN}(\mathbf{0},\sigma^2\mathbf{I}_{2M}).
\end{align}
Since the end-to-end channel is block diagonal, the optimal transmit covariance is diagonal, i.e.,
\begin{align}
    \mathbf{Q}=\diag(p_H,p_V),
    \qquad
    p_H,p_V\ge 0,
    \quad
    p_H+p_V\le P.
\end{align}
Hence, the achievable rate becomes
\begin{align}
    R_{\rm RIS}
    =
    \log_2\!\left(1+\frac{p_H}{\sigma^2}\gamma_{{\rm RIS},H}\right)
    +
    \log_2\!\left(1+\frac{p_V}{\sigma^2}\gamma_{{\rm RIS},V}\right),
\end{align}
where
\begin{align}
&  \gamma_{{\rm RIS},i}
    =
    \beta_1\beta_2
    \|\mathbf{a}_{{\rm BS},i}(\phi)\|^2
    \left|
    \mathbf{a}_{{\rm RIS},i}^{\Ttran}(\phi)
    \diag(\boldsymbol{\phi}_i)
    \mathbf{a}_{{\rm RIS},i}(\tilde\phi)
    \right|^2, \nonumber\\
    &\quad i\in\{H,V\}.
\end{align}
The optimal RIS phases align the reflected terms coherently, which yields
\begin{align}
    \gamma_{{\rm RIS},H}^{\star}
    =
    \gamma_{{\rm RIS},V}^{\star}
    \triangleq
    \gamma_{\rm RIS}
    =
    \beta_1\beta_2
    M
    N^2,
\end{align}
where we used $\|\mathbf{a}_{{\rm BS},i}(\phi)\|^2 = M$ for $i\in\{H,V\}$.

In this case, the rate expression simplifies to
\begin{align}
    R_{\rm RIS}
    =
    \log_2\!\left(1+\frac{p_H}{\sigma^2}\gamma_{\rm RIS}\right)
    +
    \log_2\!\left(1+\frac{p_V}{\sigma^2}\gamma_{\rm RIS}\right).
\end{align}
Due to symmetry, the optimal power allocation is given by equal power splitting, $
    p_H^{\star}
    =
    p_V^{\star}
    =
    \frac{P}{2}$. Substituting this into the rate expression, we obtain
\begin{align}
    R_{\rm RIS}^{\star}
    =
    2
    \log_2\!\left(
    1+
    \frac{P\beta_1\beta_2MN^2}{2\sigma^2}
    \right).
\end{align}
\vspace{-4mm}

\subsection{NCR-Assisted Communications}
\vspace{-1mm}

 The UE--NCR channel is
 \vspace{-2mm}
\begin{align}
    \mathbf{H}_1=\sqrt{\beta_1}\begin{bmatrix}e^{-\imagunit \theta_{1,H}} & 0 \\ 0 & e^{-\imagunit \theta_{1,V}}\end{bmatrix}\in\mathbb{C}^{2\times 2},
\end{align}
where $\beta_1 > 0$ denotes the large-scale channel gain of the UE--NCR link. The terms $\theta_{1,H}$ and $\theta_{1,V}$ represent the phase shifts experienced by the horizontal and vertical polarization components, respectively. The NCR--BS channel is
\begin{align}
    \mathbf{H}_2
    =
    \sqrt{\beta_2}e^{-\imagunit\theta_2}
    \begin{bmatrix}
        \mathbf{a}_{{\rm BS},H}(\phi) & \mathbf{0}\\
        \mathbf{0} & \mathbf{a}_{{\rm BS},V}(\phi)
    \end{bmatrix}\in\mathbb{C}^{2M\times 2}.
\end{align}
The amplification matrix of the NCR is represented by
\vspace{-2mm}
\begin{align}
\mathbf{A}=\diag(\alpha_H,\alpha_V),
    \qquad
    \alpha_H,\alpha_V\ge 0.
\end{align}

The received signal at the BS is
\begin{align}
    \mathbf{y}
    =
    \mathbf{H}_2\mathbf{A}\mathbf{H}_1\mathbf{s}
    +
    \mathbf{H}_2\mathbf{A}\mathbf{n}_1
    +
    \mathbf{n}_2,
\end{align}
where  $\mathbf{n}_1\sim\mathcal{CN}(\mathbf{0},\sigma^2\mathbf{I}_2),
    \qquad
    \mathbf{n}_2\sim\mathcal{CN}(\mathbf{0},\sigma^2\mathbf{I}_{2M})$.

Since both the signal and noise components are decoupled across the two polarizations, the resulting channel is block diagonal. Therefore, there is no benefit from transmitting correlated signals across the two branches, and the optimal transmit covariance matrix is diagonal, i.e., $\mathbf{Q}=\diag(p_H,p_V)$. Then, the rate can be written as \cite{demir2026ncr}
\vspace{-2mm}
\begin{align}
    R_{\rm NCR}
    =
    \log_2\!\left(1+p_H\gamma_{{\rm NCR},H}\right)
    +
    \log_2\!\left(1+p_V\gamma_{{\rm NCR},V}\right),
\end{align}
where
\vspace{-2mm}
\begin{equation}
    \gamma_{{\rm NCR},i}
    =
    \frac{\alpha_i^2\beta_1\beta_2 M}{\sigma^2(1+\alpha_i^2\beta_2M)},
    \qquad i\in\{H,V\}. \vspace{-2mm}
\end{equation}

The total NCR radiated power is
\vspace{-2mm}
\begin{align}
    P_{\rm NCR}^{\rm rad}
    =
    \alpha_H^2(p_H\beta_1+\sigma^2)
    +
    \alpha_V^2(p_V\beta_1+\sigma^2).
\end{align}
\vspace{-4mm}

\begin{lemma}\label{lem:ncr_two_stream_opt}
Consider the blocked-link NCR model with identical horizontal and
vertical polarization branches. The sum-rate maximization problem
\begin{align}
\max_{p_H,p_V,\alpha_H,\alpha_V}
\quad
\sum_{i\in\{H,V\}}
\log_2\!\bigl(1+p_i g(\alpha_i)\bigr)
\end{align}
subject to
\begin{align}
p_H+p_V \le P,
\end{align}
\begin{align}
\alpha_H^2(\beta_1 p_H+\sigma^2)+\alpha_V^2(\beta_1 p_V+\sigma^2)
\le
\bar P_{\rm NCR}^{\rm rad},
\end{align}
and $\alpha_H,\alpha_V,p_H,p_V\ge 0$, where
\begin{align}
g(\alpha)=\frac{\alpha^2\beta_1\beta_2 M}{\sigma^2(1+\alpha^2\beta_2 M)},
\end{align}
admits an optimal solution of one of the following two forms:
\vspace{-2mm}
\begin{itemize}
\item \textbf{Symmetric two-stream solution:}
\begin{align}
p_H^\star=p_V^\star=\frac{P}{2},
\qquad
\alpha_H^\star=\alpha_V^\star=\sqrt{\frac{\bar P_{\rm NCR}^{\rm rad}}{\beta_1 P+2\sigma^2}},
\end{align}
with achievable rate
\begin{align}
R_{\rm sym}
=
2\log_2\!\left(
1+
\frac{\beta_1\beta_2MP\bar P_{\rm NCR}^{\rm rad}}
{2\sigma^2(\beta_1P+2\sigma^2+\beta_2M\bar P_{\rm NCR}^{\rm rad})}
\right).
\end{align}
\vspace{-2mm}

\item \textbf{One-stream solution:}
\begin{align}
(p_H^\star,p_V^\star)
=
(P,0)
\quad \text{or} \quad
(0,P),
\end{align}
and correspondingly
\begin{align}
(\alpha_H^\star,\alpha_V^\star)
=
\left(
\sqrt{\frac{\bar P_{\rm NCR}^{\rm rad}}{\beta_1P+\sigma^2}},0
\right)
\ \text{or} \ 
\left(
0,\sqrt{\frac{\bar P_{\rm NCR}^{\rm rad}}{\beta_1P+\sigma^2}}
\right),
\end{align}
with achievable rate
\begin{align}
R_{\rm os}
=
\log_2\!\left(
1+
\frac{\beta_1\beta_2MP\bar P_{\rm NCR}^{\rm rad}}
{\sigma^2(\beta_1P+\sigma^2+\beta_2M\bar P_{\rm NCR}^{\rm rad})}
\right).
\end{align}
\end{itemize}

Hence, the globally optimal solution is given by the candidate that
yields the larger rate, i.e.,
\begin{align}
R^\star=\max\{R_{\rm sym},R_{\rm os}\}.
\end{align}

\end{lemma}
\vspace{-2mm}

\begin{proof}
Since the objective function is monotonically increasing in both $p_i$
and $u_i$, and more sensitive to increase in $p_i$, it can be shown that both the UE transmit-power constraint and the NCR
radiated-power constraint are active at the optimum. Hence, $p_H+p_V=P$
and $
\alpha_H^2(\beta_1 p_H+\sigma^2)+\alpha_V^2(\beta_1 p_V+\sigma^2)
=
\bar P_{\rm NCR}^{\rm rad}$.

Because the two polarization branches are identical, the optimization
problem is invariant to exchanging the labels $H$ and $V$. Therefore,
a symmetric stationary candidate is obtained as
\begin{align}
p_H=p_V=\frac{P}{2},
\qquad
\alpha_H=\alpha_V=\sqrt{\frac{\bar P_{\rm NCR}^{\rm rad}}{\beta_1 P+2\sigma^2}},
\end{align}
which gives the rate $R_{\rm sym}$.

Another candidate arises on the boundary, where all the transmit power
and amplification budget are assigned to a single branch. Without loss
of generality, this gives
\begin{align}
p_H=P,\quad p_V=0,
\qquad
\alpha_H=\sqrt{\frac{\bar P_{\rm NCR}^{\rm rad}}{\beta_1P+\sigma^2}},
\quad
\alpha_V=0,
\end{align}
which yields the rate $R_{\rm os}$. By symmetry, allocating all
resources to the $V$ branch gives the same rate.

Therefore, the globally optimal solution is obtained by comparing the
rates of these candidate structures and selecting the one with the
larger value, which proves the lemma.
\end{proof}

\vspace{-2mm}

\section{NCR Design Under a Joint Total-Power Budget}

As an alternative to imposing separate constraints on the UE transmit power and the NCR radiated power, we now consider a joint total-power budget that limits their sum. Specifically, we solve
\begin{subequations}
\begin{align}
\max_{\substack{p_H,p_V\ge 0\\ \alpha_H,\alpha_V\ge 0}}
\quad
&
\log_2\!\left(1+p_H\gamma_{{\rm NCR},H}\right)
+
\log_2\!\left(1+p_V\gamma_{{\rm NCR},V}\right)
\label{eq:joint_budget_obj}
\\
\text{s.t.}\quad
&
p_H+p_V\nonumber\\
&+
\alpha_H^2(\beta_1p_H+\sigma^2)
+
\alpha_V^2(\beta_1p_V+\sigma^2)
\le P_{\rm tot},
\label{eq:joint_budget_const}
\end{align}
\end{subequations}
where
\begin{align}
    \gamma_{{\rm NCR},i}
    =
    \frac{\alpha_i^2\beta_1\beta_2M}
    {\sigma^2(1+\alpha_i^2\beta_2M)},
    \qquad i\in\{H,V\}.
\end{align}

We first characterize the optimal symmetric two-stream design and then compare it with the optimal one-stream design.

\subsubsection{Optimal Symmetric Two-Stream Design}

Under the symmetric polarization model, we consider
\begin{align}
    p_H=p_V=p,
    \qquad
    \alpha_H=\alpha_V=\alpha.
\end{align}
In this case, \eqref{eq:joint_budget_const} reduces to
\begin{align}
    2p+2\alpha^2(\beta_1p+\sigma^2)\le P_{\rm tot}.
\end{align}
Since the achievable rate is increasing in $\alpha^2$ for any fixed $p$, the constraint is active at the optimum, which yields
\begin{align}
    \alpha^2
    =
    \frac{P_{\rm tot}-2p}{2(\beta_1p+\sigma^2)},
    \qquad
    0\le p\le \frac{P_{\rm tot}}{2}.
\end{align}
Substituting this into the branch SNR gives
\begin{align}
    \gamma_{\rm 2str}(p)
    =
    \frac{
    p\beta_1\beta_2M(P_{\rm tot}-2p)
    }{
    2\sigma^2\!\left(
    (\beta_1-\beta_2M)p+\sigma^2+\frac{\beta_2MP_{\rm tot}}{2}
    \right)
    }.
\end{align}
Hence, the symmetric two-stream rate is
\begin{align}
    R_{\rm 2str}(p)
    =
    2\log_2\!\left(1+\gamma_{\rm 2str}(p)\right).
\end{align}

To simplify the notation, we define
\begin{align}
    &x \triangleq 2p,\qquad
    A \triangleq P_{\rm tot},\qquad \\
    &c \triangleq \beta_1-\beta_2M,\qquad
    B_2 \triangleq 2\sigma^2+\beta_2MP_{\rm tot}.
\end{align}
Then, maximizing $R_{\rm 2str}(p)$ is equivalent to maximizing
\begin{align}
    f_2(x)
    \triangleq
    \frac{(A-x)x}{cx+B_2},
    \qquad 0\le x\le A.
\end{align}

\begin{lemma}\label{lem:joint_budget_two_stream}
The symmetric two-stream problem admits a unique interior maximizer
\begin{align}
    x_2^\star\in(0,A),
\end{align}
which is the unique positive root of
\begin{align}
    cx^2+2B_2x-AB_2=0.
\end{align}
Equivalently,
\begin{align}
    x_2^\star
    =
    \begin{cases}
    \dfrac{-B_2+\sqrt{B_2(B_2+Ac)}}{c}, & c\neq 0,\\[2ex]
    \dfrac{A}{2}, & c=0.
    \end{cases}
\end{align}
The corresponding optimal allocation is
\begin{align}
    p_H^\star=p_V^\star=\frac{x_2^\star}{2},
\end{align}
and
\begin{align}
    \alpha_H^\star=\alpha_V^\star=\alpha_2^\star,
    \qquad
    (\alpha_2^\star)^2
    =
    \frac{A-x_2^\star}{\beta_1x_2^\star+2\sigma^2}.
\end{align}
\end{lemma}

\begin{proof}
The first derivative of $f_2(x)$ is
\begin{align}
    f_2'(x)
    =
    -\frac{cx^2+2B_2x-AB_2}{(cx+B_2)^2},
\end{align}
and the second derivative is
\begin{align}
    f_2''(x)
    =
    -\frac{2B_2(Ac+B_2)}{(cx+B_2)^3}.
\end{align}
Since
\vspace{-2mm}
\begin{align}
    B_2>0,
    \qquad
    Ac+B_2
    =
    \beta_1P_{\rm tot}+2\sigma^2>0,
\end{align}
and
\vspace{-2mm}
\begin{align}
    cx+B_2
    =
    \beta_1x+2\sigma^2+\beta_2M(P_{\rm tot}-x)>0
\end{align}
for all $x\in[0,A]$, we have $f_2''(x)<0$ on the feasible interval. Hence, $f_2(x)$ is strictly concave and has a unique maximizer.

Moreover,
\begin{align}
    f_2(0)=f_2(A)=0,
\end{align}
while $f_2(x)>0$ for every $x\in(0,A)$. Therefore, the maximizer is interior and is given by the unique root of
\begin{align}
    cx^2+2B_2x-AB_2=0.
\end{align}
The expressions for $p_H^\star$, $p_V^\star$, and $(\alpha_2^\star)^2$ follow immediately.
\end{proof}

Using Lemma~\ref{lem:joint_budget_two_stream}, the optimal symmetric two-stream rate is
\begin{align}
    R_{\rm 2str}^\star
    =
    2\log_2\!\left(
    1+
    \frac{
    \beta_1\beta_2M\,x_2^\star(A-x_2^\star)
    }{
    2\sigma^2(cx_2^\star+B_2)
    }
    \right).
\end{align}

\subsubsection{Optimal One-Stream Design}

We next consider the boundary design where all resources are allocated to a single polarization branch. Without loss of generality, let
\begin{align}
    p_H=p,\qquad p_V=0,\qquad \alpha_H=\alpha,\qquad \alpha_V=0.
\end{align}
Then, \eqref{eq:joint_budget_const} becomes
\begin{align}
    p+\alpha^2(\beta_1p+\sigma^2)\le P_{\rm tot}.
\end{align}
Again, the constraint is active at the optimum, which gives
\begin{align}
    \alpha^2
    =
    \frac{P_{\rm tot}-p}{\beta_1p+\sigma^2},
    \qquad
    0\le p\le P_{\rm tot}.
\end{align}
The resulting SNR is
\begin{align}
    \gamma_{\rm 1str}(p)
    =
    \frac{
    p\beta_1\beta_2M(P_{\rm tot}-p)
    }{
    \sigma^2\!\left(
    (\beta_1-\beta_2M)p+\sigma^2+\beta_2MP_{\rm tot}
    \right)
    },
\end{align}
and the achievable rate is
\begin{align}
    R_{\rm 1str}(p)
    =
    \log_2\!\left(1+\gamma_{\rm 1str}(p)\right).
\end{align}

Define
\begin{align}
    B_1 \triangleq \sigma^2+\beta_2MP_{\rm tot}.
\end{align}
Then, maximizing $R_{\rm 1str}(p)$ is equivalent to maximizing
\begin{align}
    f_1(x)
    \triangleq
    \frac{(A-x)x}{cx+B_1},
    \qquad 0\le x\le A,
\end{align}
where $x\triangleq p$.

\begin{lemma}\label{lem:joint_budget_one_stream}
The optimal one-stream design admits a unique interior maximizer
\vspace{-2mm}
\begin{align}
    x_1^\star\in(0,A),
\end{align}
which is the unique positive root of
\vspace{-2mm}
\begin{align}
    cx^2+2B_1x-AB_1=0.
\end{align}
Equivalently,
\begin{align}
    x_1^\star
    =
    \begin{cases}
    \dfrac{-B_1+\sqrt{B_1(B_1+Ac)}}{c}, & c\neq 0,\\[2ex]
    \dfrac{A}{2}, & c=0.
    \end{cases}
\end{align}
The corresponding optimal one-stream rate is
\begin{align}
    R_{\rm 1str}^\star
    =
    \log_2\!\left(
    1+
    \frac{
    \beta_1\beta_2M\,x_1^\star(A-x_1^\star)
    }{
    \sigma^2(cx_1^\star+B_1)
    }
    \right).
\end{align}
\end{lemma}

\begin{proof}
The proof is identical to that of Lemma~\ref{lem:joint_budget_two_stream}, with $B_2$ replaced by $B_1$.
\end{proof}

Under the joint total-power budget, the globally better design is determined by directly comparing the two closed-form optimum values above.

\section{Numerical Results}

In this section, we compare the achievable-rate performance of passive RIS- and NCR-assisted communications under the analytical models developed in the previous sections. We consider a BS equipped with $M=256$ dual-polarized antennas, while the RIS employs $N\in\{128,256,512\}$ elements per polarization branch. The thermal noise power is computed over a bandwidth of $50$\,MHz with a noise figure of $9$\,dB, which gives
\vspace{-2mm}
\begin{align}
    \sigma^2 [\mathrm{dBm}]
    =
    -174 + 10\log_{10}(50\times 10^6) + 9.
\end{align}
The large-scale fading coefficients are modeled as
\vspace{-2mm}
\begin{align}
    \beta_i = \beta_0 d_i^{-\kappa},
\end{align}
where $\beta_0 = 10^{-4.2}$ is the reference channel gain at $1$\,m, $\kappa=2.5$ is the pathloss exponent, and $d_i$ denotes the corresponding propagation distance. The UE--RIS/NCR link includes a vertical height difference of $10$\,m, while the RIS/NCR--BS link is assumed to be at the same height as the BS. Hence,
\begin{align}
    d_1=\sqrt{x^2+10^2},
    \qquad
    d_2 = D_{\rm tot}-x,
\end{align}
where $x$ is the horizontal UE-to-RIS/NCR distance and $D_{\rm tot}$ is the total horizontal UE--BS separation.

For a fair comparison, the total transmit power budget is the same for all the schemes and set as $P_{\rm tot}=0.2$\,W. For the passive RIS benchmark, only the UE transmit power is constrained, since the RIS is assumed passive and does not consume radiated power. Therefore, the RIS rate is evaluated using the optimal equal power allocation across the two polarization branches. For the NCR, we consider two different power-budget models. In the first one, a joint total-power budget is imposed, such that the sum of the UE transmit power and the NCR radiated power is limited by $P_{\rm tot}$. In the second one, separate constraints are imposed on the UE transmit power and the NCR radiated power, each set to $P_{\rm tot}/2$. For both NCR models, we plot the achievable rates of the optimal two-stream and one-stream designs separately.

The performance is evaluated as a function of the horizontal position of the RIS/NCR between the UE and the BS. In particular, we consider two deployment scenarios with total horizontal distance $D_{\rm tot}=50$\,m and $D_{\rm tot}=500$\,m, respectively. These two cases represent a short-range regime and a more challenging long-range regime.\footnote{Although the deployment distances fall within the radiative near-field of very large arrays, the objective of this paper is not to characterize near-field wavefront effects but to compare RIS and NCR under an analytically tractable propagation model. Therefore, both schemes are evaluated using the same far-field array response while the actual propagation distances are reflected through the distance-dependent pathloss coefficients.} 

Fig.~\ref{fig1} shows the achievable rates versus the UE-to-RIS/NCR distance for all seven considered cases. The results reveal how the preferred technology depends strongly on the deployment geometry and the adopted power-budget model. In particular, the passive RIS benefits significantly from increasing the number of elements, due to the quadratic array gain in $N$, while the NCR can be more competitive when placed closer to the UE, where the first hop is sufficiently strong and the amplification can be effectively exploited.

In Fig.~\ref{fig1}, the passive RIS remains competitive over a wide range of positions, especially for larger values of $N$. Moreover, for the NCR, the two-stream design consistently outperforms the one-stream design, indicating that exploiting both polarization branches is beneficial in this regime. The joint power-allocation strategy further provides additional design flexibility and leads to improved performance compared to the case with separate power constraints.

In contrast, when the total distance is increased to $D_{\rm tot}=500$ m, as shown in Fig.~\ref{fig2}, the behavior changes significantly. Due to the severe cascaded pathloss affecting the RIS-assisted link, the NCR becomes the preferable solution over most of the deployment region. In fact, even the one-stream NCR design outperforms the RIS in nearly all cases, except when the UE is located extremely far from the RIS/NCR, where the first hop becomes prohibitively weak. These results highlight that, in long-range scenarios, active amplification at the repeater can effectively compensate for pathloss, whereas the passive RIS struggles despite its array gain.
\begin{figure}[t!]
        \centering
	\begin{overpic}[width=0.89\columnwidth,trim=1.8cm 0.1cm 1.5cm 0.5cm,clip,tics=10]{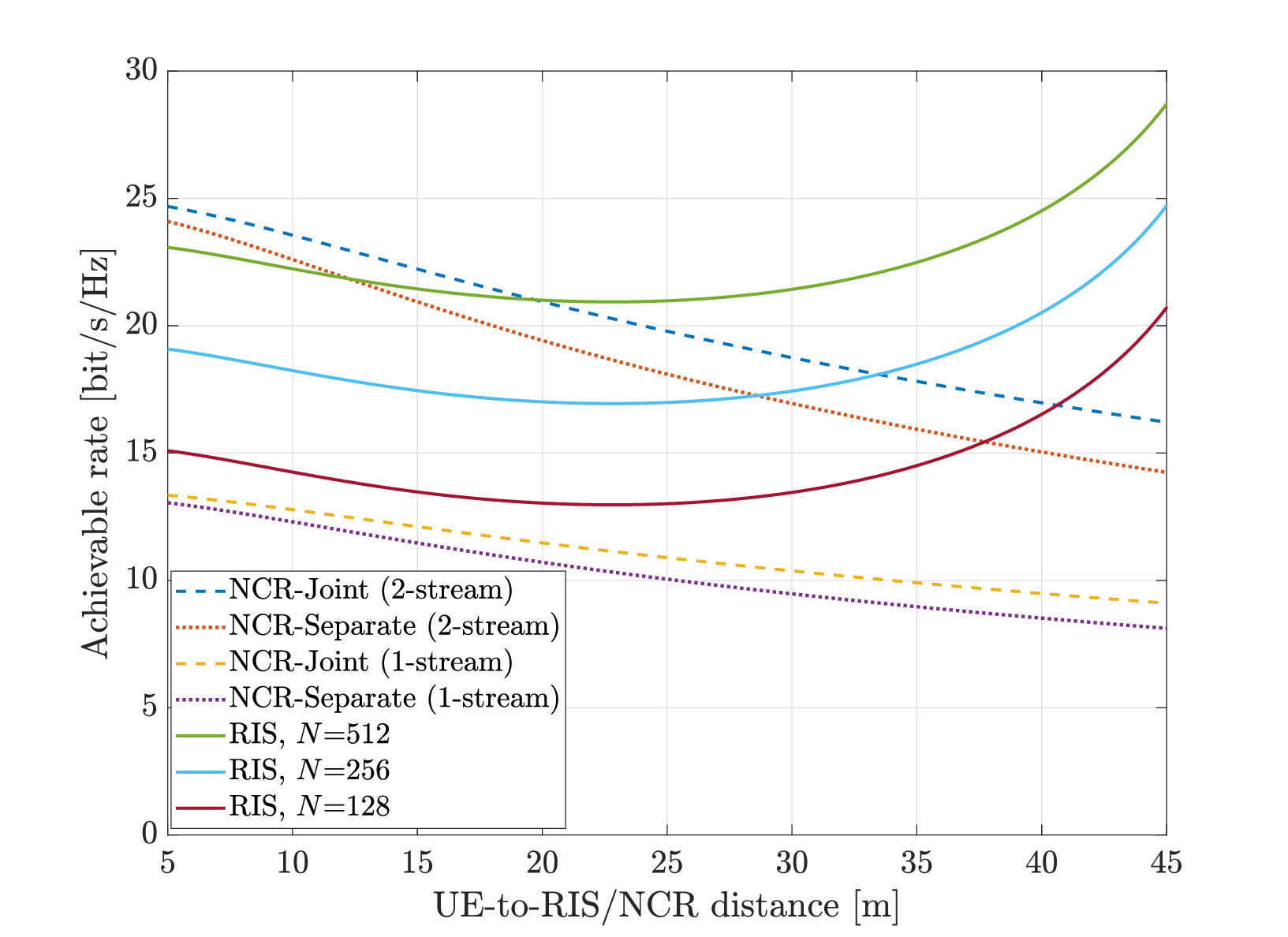}
\end{overpic} 
\vspace{-3mm}
        \caption{Achievable rate versus the UE-to-RIS/NCR distance for a total horizontal distance of $D_{\rm tot}=50$\,m.  }
        \label{fig1}
        \vspace{-6mm}
\end{figure}

\begin{figure}[t!]
        \centering
	\begin{overpic}[width=0.89\columnwidth,trim=1.8cm 0.1cm 1.5cm 0.5cm,clip,tics=10]{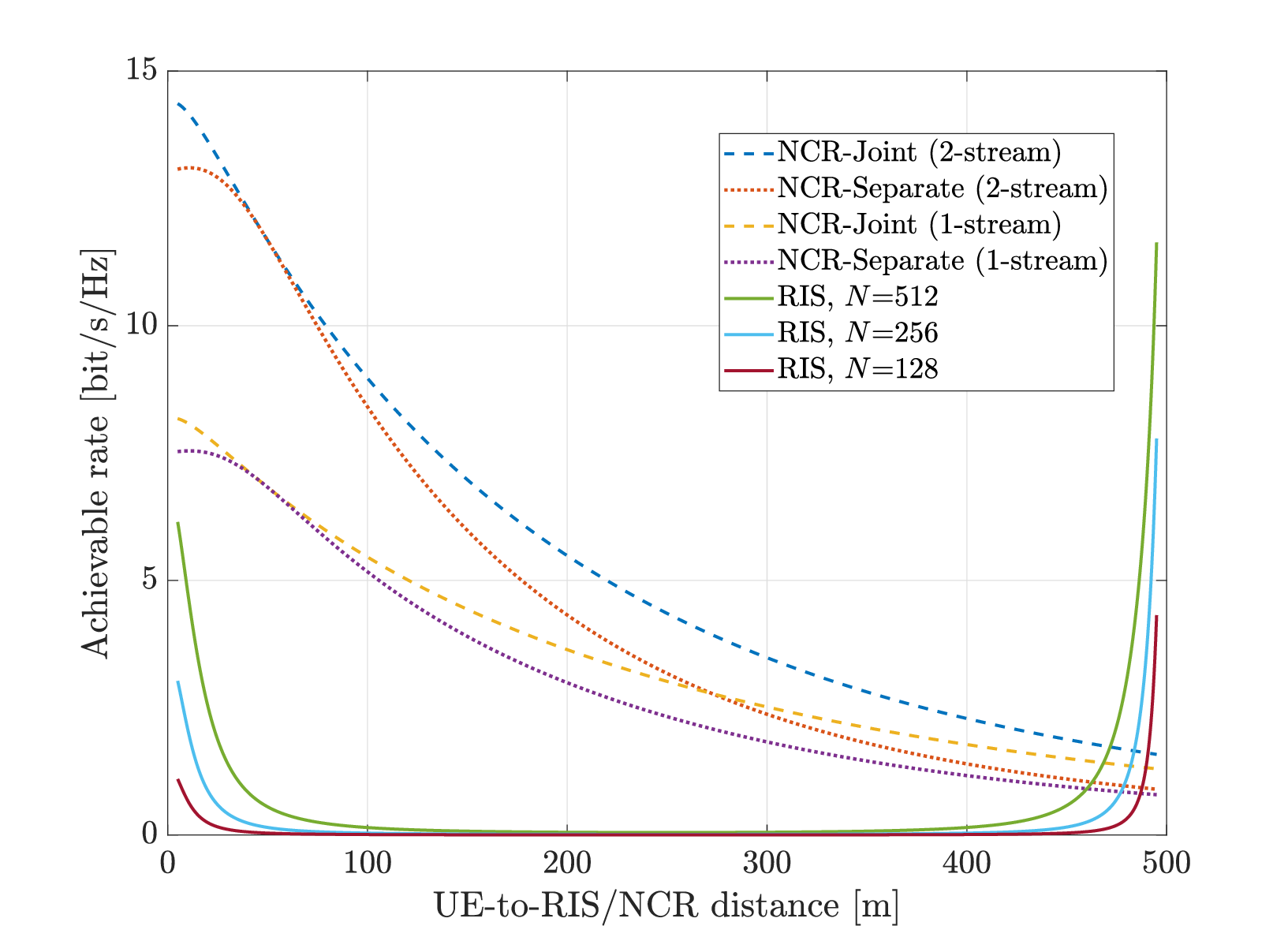}
\end{overpic} 
\vspace{-3mm}
        \caption{Achievable rate versus the UE-to-RIS/NCR distance for a total horizontal distance of $D_{\rm tot}=500$\,m. }
        \label{fig2}
        \vspace{-6mm}
\end{figure}

\vspace{-2mm}

\section{Conclusions}
\vspace{-2mm}

This paper presented a comparative analysis of dual-polarized RIS- and NCR-assisted communications under a unified framework. The results showed that while passive RISs can achieve competitive performance in short-range scenarios by exploiting large array gains, NCRs become more effective in long-range deployments due to their active amplification capability. Overall, the findings highlight that the preferable technology is highly geometry-dependent and should be selected based on the deployment conditions and system constraints.

The analysis assumes a blocked direct link, identical polarization branches, perfect channel state information, and ideal hardware. In practical deployments, polarization mismatch, channel-estimation errors, amplifier nonlinearities, and other hardware impairments may reduce the achievable gains, particularly for coherent RIS beamforming and high-gain NCR operation. 

\vspace{-2mm}

\bibliographystyle{IEEEtran}
\bibliography{IEEEabrv,refs}

\end{document}